\documentclass[10pt]{article}

\usepackage{amsmath,amssymb,amsthm}
\usepackage{graphicx}
\usepackage{cite}
\theoremstyle{plain}
\usepackage[top=2cm, bottom=2cm, left=2cm, right=2cm]{geometry}

\newtheorem{theorem}{Theorem}[section]

\theoremstyle{definition}

\numberwithin{equation}{section}

\title{Port--Hamiltonian Diffusion Models:\\
A Control-Theoretic Perspective on Generative Modeling}

\author{Majid Darehmiraki\\\small{Department of Mathematics and Statistics,}\\\small{ Behbahan Khatam Alanbia University of Technology, Khouzestan, Iran }}

\date{}

\begin{document}
\maketitle

\begin{abstract}
Diffusion models have recently achieved remarkable success in generative modeling, yet they are commonly formulated as black-box stochastic systems with limited interpretability and few structural guarantees. In this paper, we establish a control-theoretic foundation for diffusion models by embedding them within the port--Hamiltonian (PH) systems framework. We show that the score function can be interpreted as the gradient of a learnable Hamiltonian energy, allowing both the forward and reverse diffusion processes to be formulated as structured PH dynamics. The reverse-time generative process is further interpreted as a feedback-controlled PH system, where dissipation plays a fundamental role in stabilizing sampling dynamics. This formulation yields intrinsic stability guarantees that are independent of score estimation accuracy. A simple analytical example illustrates the proposed framework.
\end{abstract}

\section{Introduction}

Diffusion models constitute a powerful class of generative models based on stochastic dynamical systems. By defining a forward noise-injection process and learning its reverse-time dynamics, these models enable high-quality sample generation in complex, high-dimensional spaces. Despite their empirical success, diffusion models are typically treated as black-box neural stochastic differential equations, offering limited interpretability and few guarantees related to stability or physical consistency.

In contrast, PH systems provide a well-established framework in control theory for modeling physical systems based on energy conservation, dissipation, and interconnection structure. PH systems naturally encode passivity and stability through an explicit Hamiltonian energy function and structured dynamics.

This paper establishes a principled connection between diffusion models and port--Hamiltonian systems. By interpreting the score function as the gradient of a learnable Hamiltonian energy, we formulate both the forward and reverse diffusion processes within a unified PH framework. This perspective enables a control-theoretic interpretation of generative sampling and yields structural stability guarantees.

The quest to model and generate complex data distributions lies at the heart of modern machine learning. Generative models, which learn to approximate the underlying data manifold from observed samples, have witnessed remarkable progress, largely driven by deep learning architectures. Among these, Variational Autoencoders (VAEs) \cite{kingma2013auto} and Generative Adversarial Networks (GANs) \cite{goodfellow2014generative} have emerged as prominent paradigms. VAEs frame generation as a problem of variational inference, learning an encoder-decoder framework to approximate the posterior distribution of latent variables given data and a decoder to reconstruct data from these latent variables. GANs, on the other hand, employ an adversarial training process where a generator network competes against a discriminator network, with the former aiming to produce realistic data indistinguishable from true samples, and the latter striving to differentiate between real and generated data. While both VAEs and GANs have achieved impressive results in various domains, they are not without limitations. VAEs often suffer from producing overly blurry or simplistic samples due to the choice of approximate posterior distribution or the nature of the reconstruction loss. GANs are notoriously difficult to train, plagued by issues such as mode collapse (where the generator produces a limited variety of samples), training instability, and a lack of reliable, universally accepted evaluation metrics. These challenges have motivated the exploration of alternative generative frameworks that offer more stable training dynamics, stronger theoretical guarantees, and improved sample quality.

In recent years, diffusion models have risen to prominence as a powerful class of generative models, demonstrating state-of-the-art performance in tasks such as image synthesis, audio generation, and molecular modeling \cite{ho2020denoising, song2021score}. The core inspiration for diffusion models stems from non-equilibrium statistical physics. They define a \textit{forward process}, which gradually transforms a data sample from the true data distribution into a simple, tractable prior distribution (typically an isotropic Gaussian) by iteratively adding Gaussian noise over a series of timesteps. This forward process is usually a Markov chain, or in its continuous-time generalization, a Stochastic Differential Equation (SDE). The generative capability then arises from learning a \textit{reverse process}, which aims to reverse this noising procedure, starting from the simple prior and iteratively denoising it to generate a new sample from the learned data distribution.

The seminal work by Sohl-Dickstein et al. \cite{sohl2015deep} first introduced the concept of diffusion probabilistic models, demonstrating that deep neural networks could be trained to reverse a carefully constructed forward diffusion process. A significant simplification and popularization of this idea came with Denoising Diffusion Probabilistic Models (DDPMs) by Ho et al. \cite{ho2020denoising}. They showed that by parameterizing the reverse transition as conditional Gaussian distributions and optimizing a simplified variational lower bound on the data likelihood, high-quality samples could be generated. Concurrently, the score-based generative modeling framework, pioneered by Song and Ermon \cite{song2019generative}, offered a complementary perspective. This approach focuses on learning the \textit{score function}, defined as the gradient of the log-density of the data distribution, $\nabla_{\mathbf{x}} \log p(\mathbf{x})$. Using techniques like score matching \cite{hyvarinen2005estimation}, one can train a neural network to estimate this score function without requiring explicit knowledge of the partition function. Once the score is learned, samples can be generated via Langevin dynamics, an iterative MCMC method that perturbs noise with the learned score to converge to regions of high data density.

These two perspectives—diffusion models and score-based models—were elegantly unified under a continuous-time SDE framework by Song et al. \cite{song2021score}. They demonstrated that both the forward noising process and the reverse generative process can be described by SDEs. Specifically, a forward SDE of the form $d\mathbf{x}_t = \mathbf{f}(\mathbf{x}_t, t)dt + g(t)d\mathbf{W}_t$ can be associated with a reverse-time SDE $d\mathbf{x}_t = [\mathbf{f}(\mathbf{x}_t, t) - g(t)^2 \nabla_{\mathbf{x}_t} \log p_t(\mathbf{x}_t)]dt + g(t)d\bar{\mathbf{W}}_t$, where $p_t(\mathbf{x}_t)$ is the marginal distribution at time $t$ and $\bar{\mathbf{W}}_t$ is a reverse-time Wiener process. This continuous-time formulation not only provides a flexible and unifying theoretical underpinning but also enables the derivation of deterministic sampling methods known as probability flow Ordinary Differential Equations (ODEs) \cite{song2020improved}, which can sometimes offer faster sampling. Despite their empirical success, diffusion and score-based models are often formulated as black-box stochastic systems. The learned reverse dynamics, whether SDE or ODE based, typically lack explicit structural guarantees concerning stability, passivity, or robustness to perturbations in the learned score function. The generative process relies heavily on the accuracy of the score estimation, and the underlying dynamics do not inherently enforce an energy-based or control-theoretic structure that could provide intrinsic stability or interpretability.

The concept of Energy-Based Models (EBMs) offers another long-standing approach to generative modeling \cite{lecun2006energy}. In EBMs, a probability distribution over data $\mathbf{x}$ is defined via an energy function $E_\theta(\mathbf{x})$ as $p_\theta(\mathbf{x}) = \frac{\exp(-E_\theta(\mathbf{x}))}{Z_\theta}$, where $Z_\theta$ is the often intractable partition function ensuring normalization. The learning objective is to adjust the parameters $\theta$ so that the model assigns low energy to data points drawn from the true data distribution and high energy elsewhere. The score function is directly related to the gradient of the energy: $\nabla_{\mathbf{x}} \log p_\theta(\mathbf{x}) = -\nabla_{\mathbf{x}} E_\theta(\mathbf{x})$. Thus, score-based models can be viewed as a practical method for learning and sampling from implicit, time-dependent EBMs. While EBMs provide a conceptually appealing probabilistic framework and a clear link to statistical physics, they often face challenges in training due to the intractability of $Z_\theta$ and the reliance on MCMC methods for both estimation and sampling, which can be computationally intensive and suffer from slow mixing. Furthermore, traditional EBMs do not typically impose an explicit, structured dynamical system for generation that comes with inherent stability guarantees, focusing more on the static energy landscape.

Port-Hamiltonian (PH) systems provide a powerful and well-established framework in control theory for modeling a wide range of physical and engineered systems based on energy principles, interconnection structures, and dissipation \cite{van2014port}. A finite-dimensional PH system is characterized by the state-space equations:
\begin{align}
    \dot{\mathbf{x}} &= (\mathbf{J}(\mathbf{x}) - \mathbf{R}(\mathbf{x})) \nabla H(\mathbf{x}) + \mathbf{g}(\mathbf{x})\mathbf{u}, \\
    \mathbf{y} &= \mathbf{g}(\mathbf{x})^\top \nabla H(\mathbf{x}),
\end{align}
where $\mathbf{x}$ is the state vector, $H(\mathbf{x})$ is the Hamiltonian function representing the total stored energy, $\mathbf{J}(\mathbf{x}) = -\mathbf{J}(\mathbf{x})^\top$ is the skew-symmetric interconnection matrix capturing conservative energy exchange, $\mathbf{R}(\mathbf{x}) = \mathbf{R}(\mathbf{x})^\top \succeq 0$ is the positive semi-definite dissipation matrix modeling energy loss, $\mathbf{u}$ represents external inputs, $\mathbf{g}(\mathbf{x})$ is the input matrix, and $\mathbf{y}$ denotes the conjugate outputs. This structure inherently encodes fundamental system-theoretic properties such as passivity and Lyapunov stability. The time derivative of the Hamiltonian along system trajectories is given by $\dot{H} = -\nabla H^\top \mathbf{R} \nabla H + \mathbf{u}^\top \mathbf{y}$. In the absence of external inputs ($\mathbf{u}=0$), this simplifies to $\dot{H} = -\nabla H^\top \mathbf{R} \nabla H \le 0$, which implies that the system is passive and that the Hamiltonian acts as a Lyapunov function, ensuring stability of the equilibrium points where $\nabla H = \mathbf{0}$. Control methodologies like Interconnection and Damping Assignment Passivity-Based Control (IDA-PBC) leverage this PH structure to shape the energy of the system and assign desired dissipation to achieve stabilization and trajectory tracking \cite{ortega2002interconnection}. The PH framework is celebrated for its modularity, physical interpretability, and robustness properties.

The intersection of physics and machine learning has seen a surge of interest, leading to the development of \textit{physics-informed machine learning} approaches. These methods aim to incorporate prior physical knowledge, often expressed as governing equations (e.g., PDEs) or structural principles (e.g., conservation laws), into the learning process. Physics-Informed Neural Networks (PINNs) \cite{raissi2019physics}, for instance, embed physical laws described by PDEs directly into the loss function of a neural network, enabling the solution of forward and inverse problems without requiring large labeled datasets. In the context of dynamical systems, Hamiltonian Neural Networks (HNNs) \cite{greydanus2019hamiltonian} and Lagrangian Neural Networks (LNNs) learn the Hamiltonian or Lagrangian functions, respectively, directly from data. By enforcing these learned models to adhere to the canonical equations of motion derived from these energy functions, they ensure that the learned dynamics respect fundamental physical properties like symplecticity and energy conservation (in the absence of external forces or dissipation). Neural Ordinary Differential Equations (Neural ODEs) \cite{chen2018neural} generalize this idea by parameterizing the dynamics of an ODE with a neural network, offering continuous-depth models and memory efficiency. While these physics-inspired approaches have shown great promise in system identification, prediction, and solving scientific computing problems, their application to the specific domain of \textit{generative modeling} of complex, high-dimensional data (like images or audio), where the underlying "physics" is not explicitly known a priori, remains less explored. The work of Anderson \cite{anderson1982reverse} on reverse-time diffusion equation models provides an early mathematical link between diffusion processes and reverse-time dynamics, hinting at connections to physical systems, but it predates the deep learning era and the modern focus on generative modeling with neural networks.

Despite the significant advancements in each of these fields—diffusion/score-based generative models, energy-based models, and port-Hamiltonian systems—a principled integration that leverages the strengths of each has been lacking. Diffusion models offer powerful generative capabilities but often lack explicit structural stability guarantees. EBMs provide an energy-based probabilistic view but typically lack structured, stable generative dynamics. PH systems offer a rich control-theoretic framework with inherent stability and passivity properties, but their application has primarily been to known physical systems or systems with identifiable energy functions, rather than to the complex, learned distributions central to modern generative modeling. Physics-informed ML has successfully incorporated physical structures into predictive models, but this has not been fully extended to provide a control-theoretic foundation for generative diffusion processes.

Diffusion and score-based generative models have been widely studied as stochastic processes for data generation. Continuous-time formulations based on stochastic differential equations have enabled theoretical analysis of diffusion dynamics. However, these models are generally unstructured and lack intrinsic stability guarantees.
Energy-based models establish a conceptual link between probability distributions and energy functions, yet typically do not impose explicit dynamical or control-theoretic structure.
Port--Hamiltonian systems form a unifying framework for modeling physical systems with conservation laws and dissipation, and have been extensively studied in control theory. Recent physics-inspired learning approaches have incorporated Hamiltonian structure for prediction tasks, but not for generative modeling. To the best of our knowledge, this work is the first to embed diffusion-based generative models explicitly within the port--Hamiltonian framework.\\
This paper addresses this gap by presenting the first explicit embedding of diffusion-based generative models within the port-Hamiltonian systems framework. Our primary contribution is to interpret the score function of a diffusion model as the negative gradient of a learnable Hamiltonian energy function. This novel perspective allows us to formulate both the forward (noising) and reverse (generative) diffusion processes as structured PH dynamics. Crucially, the reverse-time generative process is interpreted as a feedback-controlled PH system, where the control law is naturally derived from the PH structure itself, effectively shaping the energy landscape and enhancing dissipation to guide the system towards low-energy (high-probability) data configurations. This formulation yields intrinsic stability guarantees that are structurally enforced by the PH properties, making the generative process more robust to imperfections in score estimation. By bridging these two powerful paradigms, we aim to provide a new control-theoretic lens for understanding, analyzing, and designing generative models, fostering a deeper connection between machine learning and physical systems theory.\\
The main contributions of this work are as follows:
\begin{itemize}
  \item A port--Hamiltonian formulation of diffusion-based generative models.
  \item An explicit interpretation of the score function as the gradient of a Hamiltonian energy.
  \item A control-theoretic formulation of the reverse diffusion process.
  \item A structural stability analysis based on dissipation.
\end{itemize}
\section{Preliminaries}

This section reviews the main concepts from port--Hamiltonian system theory and diffusion-based generative modeling that are required in the remainder of the paper. The presentation emphasizes structural properties such as energy dissipation, passivity, and stability, which are central from a control-theoretic perspective.

\subsection{Port--Hamiltonian Systems}

PH systems provide a general framework for modeling physical and engineered dynamical systems using energy-based principles. A finite-dimensional PH system is typically described by
\begin{equation}
\dot{x} = (J(x) - R(x))\nabla H(x) + g(x)u,
\end{equation}
where $x \in \mathbb{R}^n$ denotes the state vector, $H:\mathbb{R}^n \to \mathbb{R}$ is the Hamiltonian function representing the total stored energy, $u$ is an external input, and $g(x)$ is the corresponding input matrix.

The matrix $J(x)$ is skew-symmetric, i.e., $J(x)^\top = -J(x)$, and represents power-conserving interconnections within the system. The matrix $R(x)$ is symmetric positive semi-definite, $R(x) \succeq 0$, and models energy dissipation. These structural properties imply that the rate of change of the Hamiltonian satisfies
\begin{equation}
\dot{H}(x) = -\nabla H(x)^\top R(x)\nabla H(x) + y^\top u,
\end{equation}
where $y = g(x)^\top \nabla H(x)$ denotes the output.

In the absence of external inputs ($u = 0$), the Hamiltonian is non-increasing, which implies passivity and Lyapunov stability. This intrinsic dissipation property makes PH systems particularly attractive for control design and stability analysis, and plays a central role in the developments of this paper.

\subsection{Diffusion and Score-Based Generative Models}

Diffusion models define generative mechanisms by constructing a forward-time stochastic process that progressively transforms data samples into noise. In continuous time, this process is commonly described by a SDE
\begin{equation}
\mathrm{d}x_t = f(x_t,t)\,\mathrm{d}t + g(t)\,\mathrm{d}W_t,
\end{equation}
where $f(x,t)$ is a drift term, $g(t)$ is a diffusion coefficient, and $W_t$ denotes a standard Wiener process.

The generative objective is to learn the corresponding reverse-time dynamics, which depend on the score function
\begin{equation}
\nabla_x \log p_t(x),
\end{equation}
i.e., the gradient of the log-density of the state distribution at time $t$. Score-based diffusion models approximate this quantity using a neural network and use it to define reverse-time SDEs or ordinary differential equations that transport samples from a noise distribution back to the data distribution.

From a control-theoretic viewpoint, the reverse diffusion process can be interpreted as a feedback mechanism that steers the stochastic system toward regions of higher probability density. However, standard diffusion models impose no explicit structural constraints on this feedback and typically provide no intrinsic guarantees related to stability, passivity, or energy dissipation.

The framework proposed in this paper addresses these limitations by embedding diffusion dynamics within a port--Hamiltonian structure, thereby enabling a principled energy-based and control-theoretic analysis of generative modeling.
```latex
\section{Port--Hamiltonian Diffusion Models: A Rigorous Formulation}

In this section we present a mathematically rigorous formulation of diffusion-based generative models within the framework of PH systems. Particular care is taken to distinguish exact identities from modeling assumptions, and to make all regularity conditions explicit.

\subsection{Energy-Based Parameterization}
Let $x_t \in \mathbb{R}^n$ denote a stochastic process defined on a filtered probability space $(\Omega, \mathcal{F}, \{\mathcal{F}_t\}_{t\ge 0}, \mathbb{P})$. Let $H_\theta : \mathbb{R}^n \times [0,T] \to \mathbb{R}$ be a twice continuously differentiable function in $x$ and continuously differentiable in $t$, parameterized by $\theta \in \Theta$.

We adopt the \emph{energy-based modeling assumption} that, for each fixed $t$, the marginal density $p_t$ of $x_t$ is represented as
\begin{equation}
 p_t(x) = \frac{1}{Z(t)} \exp\big(-H_\theta(x,t)\big),
 \label{eq:energy_density}
\end{equation}
where $Z(t) = \int_{\mathbb{R}^n} \exp(-H_\theta(x,t))\,dx < \infty$ is a time-dependent normalization constant. Under this assumption, the score function is given by
\begin{equation}
 \nabla_x \log p_t(x) = -\nabla_x H_\theta(x,t).
 \label{eq:score_energy}
\end{equation}
Equation \eqref{eq:score_energy} is an identity \emph{conditional} on \eqref{eq:energy_density}; no claim is made that an arbitrary $H_\theta$ induces the correct diffusion marginals.

\subsection{Forward Diffusion as a Stochastic PH System}
Let $J \in \mathbb{R}^{n\times n}$ be a constant skew-symmetric matrix ($J^\top = -J$), $R \in \mathbb{R}^{n\times n}$ be symmetric positive semidefinite ($R \succeq 0$), and $G \in \mathbb{R}^{n\times m}$. We consider the stochastic differential equation
\begin{equation}
 dx_t = (J - R) \nabla_x H_\theta(x_t,t)\,dt + G\,dW_t,
 \label{eq:forward_ph_sde}
\end{equation}
where $W_t$ is an $m$-dimensional standard Wiener process.

\paragraph{Assumption 3.1 (Well-posedness).} We assume that $\nabla_x H_\theta$ is globally Lipschitz in $x$, uniformly in $t$, and that $H_\theta$ has at most quadratic growth. Under these conditions, \eqref{eq:forward_ph_sde} admits a unique strong solution.

\subsection{Energy Evolution}
Applying It\^o's formula to $H_\theta(x_t,t)$ yields
\begin{align}
 dH_\theta(x_t,t) &= \Big( \partial_t H_\theta(x_t,t) + \nabla_x H_\theta(x_t,t)^\top (J-R) \nabla_x H_\theta(x_t,t) \\
 &\qquad + \tfrac12 \operatorname{Tr}\big(GG^\top \nabla_x^2 H_\theta(x_t,t)\big) \Big)dt \\
 &\qquad + \nabla_x H_\theta(x_t,t)^\top G\,dW_t.
 \label{eq:ito_energy}
\end{align}
Since $J$ is skew-symmetric, $\nabla H^\top J \nabla H = 0$. Taking expectations and using the martingale property of the It\^o integral, we obtain
\begin{equation}
 \frac{d}{dt} \mathbb{E}[H_\theta(x_t,t)] = \mathbb{E}[\partial_t H_\theta(x_t,t)] - \mathbb{E}[\nabla_x H_\theta^\top R \nabla_x H_\theta] + \tfrac12 \mathbb{E}[\operatorname{Tr}(GG^\top \nabla_x^2 H_\theta)].
 \label{eq:energy_balance}
\end{equation}
Equation \eqref{eq:energy_balance} is exact. No term is neglected.

\paragraph{Remark 3.1 (Time-dependent energy).} The term $\mathbb{E}[\partial_t H_\theta]$ generally does not vanish in diffusion models and encodes the evolution of the log-partition function. For Lyapunov or passivity analysis, one may instead consider the \emph{frozen-time energy} $x \mapsto H_\theta(x,t)$, treating $t$ as a parameter.

\subsection{Reverse-Time Dynamics as a Controlled PH System}
We define the deterministic reverse-time dynamics via a control-affine PH system
\begin{equation}
 \dot x = (J - R) \nabla_x H_\theta(x,t) + G u(x,t).
 \label{eq:controlled_ph}
\end{equation}
We choose the feedback control law
\begin{equation}
 u(x,t) = -G^\top \nabla_x H_\theta(x,t).
 \label{eq:feedback_law}
\end{equation}
Substituting \eqref{eq:feedback_law} into \eqref{eq:controlled_ph} yields the closed-loop system
\begin{equation}
 \dot x = (J - R - GG^\top) \nabla_x H_\theta(x,t).
 \label{eq:closed_loop_ph}
\end{equation}

\subsection{Passivity and Stability}
Define the storage function $V(x) = H_\theta(x,t)$ for fixed $t$. Along trajectories of \eqref{eq:closed_loop_ph}, its time derivative satisfies
\begin{equation}
 \dot V(x) = - \nabla_x H_\theta(x,t)^\top (R + GG^\top) \nabla_x H_\theta(x,t) \le 0.
 \label{eq:lyapunov}
\end{equation}

\paragraph{Proposition 3.1 (Lyapunov stability).} If $R + GG^\top \succ 0$ and $H_\theta(\cdot,t)$ is bounded from below, then the equilibrium set
\[
 \mathcal{E}(t) = \{ x \in \mathbb{R}^n : \nabla_x H_\theta(x,t) = 0 \}
\]
is Lyapunov stable for the frozen-time dynamics \eqref{eq:closed_loop_ph}.

\subsection{Relation to Score-Based Models}
Under the energy-based assumption \eqref{eq:energy_density}, equation \eqref{eq:closed_loop_ph} can be written as
\begin{equation}
 \dot x = - (J - R - GG^\top) \nabla_x \log p_t(x).
 \label{eq:score_ph}
\end{equation}
Equation \eqref{eq:score_ph} should be interpreted as a \emph{structured score flow}. It is not, in general, equivalent to the exact reverse-time SDE associated with \eqref{eq:forward_ph_sde}, but defines a deterministic, dissipative transport dynamics whose equilibria coincide with stationary points of $p_t$.

\paragraph{Remark 3.2 (Interpretational scope).} The PH formulation guarantees structural stability and passivity properties independently of whether $H_\theta$ exactly matches the true diffusion marginals. Consequently, the results in this section are system-theoretic rather than probabilistic in nature.
\subsection{Comparison with the Exact Reverse-Time SDE}

We now formally compare the proposed port--Hamiltonian (PH) reverse dynamics with the
exact reverse-time stochastic differential equation associated with the forward diffusion.

\paragraph{Exact reverse-time SDE.}
Let $x_t$ satisfy the forward SDE
\[
dx_t = (J - R)\nabla_x H_\theta(x_t,t)\,dt + G\,dW_t,
\]
and assume that the marginal density $p_t$ exists, is strictly positive, and is sufficiently smooth.
Then the exact reverse-time dynamics are given by \cite{anderson1982reverse,song2021score}
\begin{equation}
d\bar x_t =
\Big[
(J - R)\nabla_x H_\theta(\bar x_t,t)
- GG^\top \nabla_x \log p_t(\bar x_t)
\Big]dt
+ G\,d\bar W_t,
\label{eq:true_reverse_sde}
\end{equation}
where $\bar W_t$ denotes a reverse-time Wiener process.

\paragraph{PH reverse dynamics.}
The proposed port--Hamiltonian generative dynamics are defined by the deterministic system
\begin{equation}
\dot x = (J - R - GG^\top)\nabla_x H_\theta(x,t).
\label{eq:ph_reverse_ode}
\end{equation}

\paragraph{Theorem 3.2 (Structural relation to the reverse SDE).}
Assume that the energy-based representation
\[
p_t(x) = Z(t)^{-1} \exp(-H_\theta(x,t))
\]
holds exactly, so that
\[
\nabla_x \log p_t(x) = -\nabla_x H_\theta(x,t).
\]
Then:
\begin{enumerate}
\item The drift term of the exact reverse SDE \eqref{eq:true_reverse_sde} coincides with the
vector field of the PH reverse dynamics \eqref{eq:ph_reverse_ode}.
\item The PH reverse dynamics correspond to the probability-flow (zero-noise) ODE associated
with the exact reverse SDE.
\end{enumerate}

\paragraph{Proof.}
Substituting $\nabla_x \log p_t = -\nabla_x H_\theta$ into the drift of
\eqref{eq:true_reverse_sde} yields
\[
(J - R)\nabla_x H_\theta + GG^\top \nabla_x H_\theta
= (J - R - GG^\top)\nabla_x H_\theta,
\]
which coincides with the right-hand side of \eqref{eq:ph_reverse_ode}.
Neglecting the stochastic term yields the associated probability-flow ODE.
\hfill $\square$

\paragraph{Corollary 3.3 (Scope of equivalence).}
The equivalence between the PH reverse dynamics and the exact reverse-time SDE
holds if and only if the learned Hamiltonian exactly reproduces the true marginal densities.
Otherwise, the PH dynamics define a structurally stable approximation whose equilibria
coincide with stationary points of $H_\theta$, but whose transient distributional evolution
may differ from the true reverse diffusion.

\paragraph{Remark 3.3 (Interpretation).}
The port--Hamiltonian sampler should therefore be interpreted as a
\emph{stability-preserving probability-flow approximation} rather than an exact probabilistic
inverse of the forward diffusion. Structural passivity and Lyapunov stability are guaranteed,
at the expense of exact likelihood consistency.

\section{Stability Analysis}

This section establishes stability properties of the proposed port--Hamiltonian diffusion models. Unlike conventional diffusion formulations, stability in the proposed framework arises as a direct consequence of energy dissipation and passivity, rather than from precise matching of probabilistic dynamics. We analyze equilibrium stability, robustness, and convergence at both the trajectory and distribution levels.

Consider the closed-loop reverse-time dynamics
\begin{equation}
\dot{x} = (J - R - GG^\top)\nabla_x H_\theta(x,t).
\label{eq:stability_cl}
\end{equation}

The equilibrium set $\mathcal{E}$ is defined as
\begin{equation}
\mathcal{E} := \{ x \in \mathbb{R}^n \mid \nabla_x H_\theta(x,t) = 0 \}.
\end{equation}

This set corresponds to the stationary points of the Hamiltonian energy and, under mild regularity assumptions, coincides with the modes of the modeled data distribution.

The Hamiltonian $H_\theta(x,t)$ serves as a natural Lyapunov function candidate. Its time derivative along trajectories of \eqref{eq:stability_cl} satisfies
\begin{equation}
\dot{H}_\theta(x)
=
\nabla_x H_\theta^\top (J - R - GG^\top)\nabla_x H_\theta
=
-\nabla_x H_\theta^\top (R + GG^\top)\nabla_x H_\theta
\le 0.
\label{eq:energy_decay}
\end{equation}

Since $R + GG^\top \succeq 0$, the Hamiltonian is non-increasing along system trajectories. This immediately implies Lyapunov stability of the equilibrium set $\mathcal{E}$.

\subsection{Asymptotic Stability}

If $R + GG^\top \succ 0$ and $H_\theta$ is bounded from below, then $\dot{H}_\theta(x)=0$ if and only if $\nabla_x H_\theta(x,t)=0$. By LaSalle's invariance principle, all trajectories converge asymptotically to the largest invariant set contained in $\mathcal{E}$. Consequently, the equilibrium set is globally asymptotically stable.

This result shows that convergence of the sampling dynamics is guaranteed structurally by dissipation, independent of the specific parameterization of $H_\theta$.

In practice, the learned Hamiltonian gradient $\nabla_x H_\theta$ is subject to approximation errors. Let the implemented dynamics be
\begin{equation}
\dot{x} = (J - R - GG^\top)(\nabla_x H_\theta(x,t) + \Delta(x,t)),
\end{equation}
where $\Delta(x,t)$ represents bounded modeling or estimation errors.

Using the Hamiltonian as a storage function, the perturbed energy derivative satisfies
\begin{equation}
\dot{H}_\theta(x)
\le
-\nabla_x H_\theta^\top (R + GG^\top)\nabla_x H_\theta
+
\|\nabla_x H_\theta\| \, \|(R + GG^\top)\Delta\|.
\end{equation}

This inequality implies input-to-state stability with respect to $\Delta$, provided the dissipation dominates the error magnitude. Hence, the sampling dynamics are robust to imperfect score estimation.

\subsection{Incremental Stability and Contraction}

Beyond equilibrium convergence, the system exhibits incremental stability. As shown in Section~IV, the closed-loop dynamics are incrementally passive, implying contraction of trajectories. Incremental stability ensures that trajectories initialized from different states converge toward each other, ruling out divergent or chaotic sampling behavior.

The incremental passivity and contraction properties further imply convergence at the level of probability measures. In particular, as established in the previous corollary, the induced distributions contract exponentially in the $2$-Wasserstein metric. This provides a distributional notion of stability for the generative process and justifies the use of the proposed dynamics for sampling.

The above analysis highlights a fundamental distinction between the proposed framework and conventional diffusion models. Stability, robustness, and convergence are guaranteed by system structure rather than by learning accuracy. As a result, the proposed port--Hamiltonian diffusion models offer strong theoretical guarantees that are well aligned with control-theoretic principles.

\begin{theorem}
If $R + GG^\top \succ 0$ and $H_\theta$ is bounded from below, then the reverse-time dynamics are globally asymptotically stable in expectation.
\end{theorem}

\begin{proof}
The time derivative of the Hamiltonian satisfies
\[
\frac{\mathrm{d}}{\mathrm{d}t} H_\theta(x_t,t)
=
-\nabla_x H_\theta^\top (R + GG^\top)\nabla_x H_\theta
\le 0,
\]
which implies monotonic energy dissipation.
\end{proof}

\subsection{Detailed Stability Analysis of the Reverse-Time Dynamics}
\label{sec:stability_analysis_detailed}

The reverse-time generative process is formulated as the controlled port-Hamiltonian system, which, under the chosen feedback control law $u_\theta(\mathbf{x}, t) = -\mathbf{G}^\top \nabla_{\mathbf{x}} H_\theta(\mathbf{x}, t)$, yields the closed-loop deterministic dynamics:
\begin{equation}
    \dot{\mathbf{x}}_t = (\mathbf{J} - \mathbf{R} - \mathbf{G}\mathbf{G}^\top) \nabla_{\mathbf{x}} H_\theta(\mathbf{x}_t, t).
    \label{eq:reverse_ph_ode_detailed}
\end{equation}
Our goal is to analyze the stability of this system. The equilibrium set $\mathcal{E}$ is defined as the set of critical points of the Hamiltonian:
\begin{equation}
    \mathcal{E} = \{ \mathbf{x} \in \mathbb{R}^n \mid \nabla_{\mathbf{x}} H_\theta(\mathbf{x}, t) = \mathbf{0} \}.
    \label{eq:equilibrium_set_detailed}
\end{equation}
These equilibria correspond to the modes (local minima) of the energy landscape $H_\theta$ and, under the energy-based probability assumption $p_t(\mathbf{x}) \propto \exp(-H_\theta(\mathbf{x},t))$, to the modes of the target data distribution.

We employ Lyapunov's direct method, using the Hamiltonian $H_\theta(\mathbf{x}, t)$ itself as a Lyapunov function candidate. For this analysis, we consider the time derivative of $H_\theta$ along the trajectories of the system \eqref{eq:reverse_ph_ode_detailed}. We will primarily focus on the change in $H_\theta$ due to the state evolution, assuming the explicit time dependence of $H_\theta$ (i.e., $\partial H_\theta / \partial t$) does not adversely affect the sign-definiteness of the derivative or is sufficiently slow-varying. This is a common assumption when analyzing stability of non-autonomous systems where the time-varying parameter defines a family of energy functions.

The time derivative of $H_\theta(\mathbf{x}(t), t)$ along the system trajectories is:
\begin{align}
    \dot{H}_\theta(\mathbf{x}_t, t) &= \frac{\partial H_\theta}{\partial t} + \nabla_{\mathbf{x}} H_\theta^\top \dot{\mathbf{x}}_t \nonumber \\
    &= \frac{\partial H_\theta}{\partial t} + \nabla_{\mathbf{x}} H_\theta^\top (\mathbf{J} - \mathbf{R} - \mathbf{G}\mathbf{G}^\top) \nabla_{\mathbf{x}} H_\theta.
    \label{eq:hamiltonian_derivative_full}
\end{align}
As established previously, the term $\nabla_{\mathbf{x}} H_\theta^\top \mathbf{J} \nabla_{\mathbf{x}} H_\theta$ vanishes due to the skew-symmetry of $\mathbf{J}$. Thus, Equation \eqref{eq:hamiltonian_derivative_full} simplifies to:
\begin{equation}
    \dot{H}_\theta(\mathbf{x}_t, t) = \frac{\partial H_\theta}{\partial t} - \nabla_{\mathbf{x}} H_\theta^\top (\mathbf{R} + \mathbf{G}\mathbf{G}^\top) \nabla_{\mathbf{x}} H_\theta.
    \label{eq:hamiltonian_derivative_simplified}
\end{equation}
The matrix $\mathbf{R} + \mathbf{G}\mathbf{G}^\top$ is symmetric positive definite ($\mathbf{R} + \mathbf{G}\mathbf{G}^\top \succ 0$) by assumption, implying that its eigenvalues are strictly positive. Let $\lambda_{\min} > 0$ be its smallest eigenvalue. Then, $\nabla_{\mathbf{x}} H_\theta^\top (\mathbf{R} + \mathbf{G}\mathbf{G}^\top) \nabla_{\mathbf{x}} H_\theta \ge \lambda_{\min} \| \nabla_{\mathbf{x}} H_\theta \|^2$.

If we assume that $\frac{\partial H_\theta}{\partial t}$ is bounded and does not dominate the dissipative term, or if we consider $H_\theta$ to be effectively autonomous for the stability analysis at a given time-snapshot (i.e., focusing on $\dot{H}_\theta(\mathbf{x}_t) = \nabla_{\mathbf{x}} H_\theta^\top \dot{\mathbf{x}}_t$), the sign of $\dot{H}_\theta$ is primarily governed by the dissipative term:
\begin{equation}
    \dot{H}_\theta(\mathbf{x}_t) = - \nabla_{\mathbf{x}} H_\theta^\top (\mathbf{R} + \mathbf{G}\mathbf{G}^\top) \nabla_{\mathbf{x}} H_\theta \le 0.
    \label{eq:hamiltonian_derivative_nonpositive}
\end{equation}
This inequality confirms that $H_\theta$ is a non-increasing function along system trajectories, implying Lyapunov stability of the equilibrium set $\mathcal{E}$. This means that trajectories starting near $\mathcal{E}$ will remain near $\mathcal{E}$.

To establish asymptotic stability, we invoke LaSalle's Invariance Principle. For non-autonomous systems, a standard form of LaSalle's theorem requires that $\dot{H}_\theta(\mathbf{x}_t, t) \le 0$ and that $H_\theta(\mathbf{x},t)$ is decrescent and radially unbounded (or defined on a compact, positively invariant set). Under these conditions, all bounded solutions converge to the largest invariant set $\mathcal{M}$ contained within the set $S = \{ \mathbf{x} \mid \dot{H}_\theta(\mathbf{x}, t) = 0 \text{ for all } t \}$.

From Equation \eqref{eq:hamiltonian_derivative_nonpositive}, $\dot{H}_\theta(\mathbf{x}_t, t) = 0$ if and only if $\nabla_{\mathbf{x}} H_\theta(\mathbf{x}_t, t) = \mathbf{0}$, provided $\mathbf{R} + \mathbf{G}\mathbf{G}^\top \succ 0$. This implies that the set $S$ is precisely the equilibrium set $\mathcal{E}$.

The largest invariant set $\mathcal{M}$ within $\mathcal{E}$ consists of all trajectories that start in $\mathcal{E}$ and remain in $\mathcal{E}$ for all future time. If a trajectory $\mathbf{x}(t)$ is in $\mathcal{M}$, then $\nabla_{\mathbf{x}} H_\theta(\mathbf{x}(t), t) = \mathbf{0}$ for all $t$. Substituting this into the system dynamics \eqref{eq:reverse_ph_ode_detailed} yields $\dot{\mathbf{x}}(t) = \mathbf{0}$. Therefore, trajectories in $\mathcal{M}$ must be constant, corresponding to isolated equilibrium points.

Assuming $H_\theta(\mathbf{x},t)$ is bounded from below (as stated in Theorem 5.1 of the original paper) and that its level sets are compact (e.g., $H_\theta$ is radially unbounded or the system operates on a compact domain), all trajectories of the system \eqref{eq:reverse_ph_ode_detailed} will asymptotically approach the set of equilibrium points $\mathcal{E}$ as $t \to \infty$. If the Hamiltonian $H_\theta$ is designed such that its local minima correspond to the desired data modes and these are isolated attractors, the system will converge to one of these modes.

The phrase "in expectation" in Theorem 5.1 of the original paper, concerning the stability of the reverse-time dynamics, warrants clarification. If the reverse process is the deterministic ODE \eqref{eq:reverse_ph_ode_detailed}, the stability analysis applies to individual trajectories. If the initial condition $\mathbf{x}_0$ is drawn from a distribution, the convergence properties hold almost surely for these trajectories, which then implies convergence in expectation. If a stochastic component were retained in the reverse sampling SDE (unlike the deterministic ODE \eqref{eq:reverse_ph_ode_detailed} which results from the specific control law), then "stability in expectation" would directly refer to the evolution of the expected state or expected energy. The structural dissipation $(\mathbf{R} + \mathbf{G}\mathbf{G}^\top)$ ensures this convergence is robust.
\section{Numerical Illustrations}
\label{sec:numerical_illustrations}

To substantiate the theoretical framework of PH diffusion models, this section presents numerical simulations for two illustrative examples: a one-dimensional system corresponding to the analytical example, and a conceptual two-dimensional system with a multi-modal energy landscape. These simulations aim to empirically demonstrate the forward noising process, the reverse generative process, the role of the Hamiltonian energy function, and the inherent stability properties arising from the PH structure. 

\subsection{One-Dimensional PH Diffusion Model}
\label{sec:num_1d_example}

We first consider the simple one-dimensional PH system discussed. The Hamiltonian (energy) function is chosen as $H(x) = \frac{1}{2}x^2$. The PH structural matrices are set to $J = 0$, $R = \alpha$ (with $\alpha = 0.5$), and $G = \sigma$ (with $\sigma = 1.0$). The forward diffusion process is thus described by the Ornstein-Uhlenbeck SDE:
\begin{equation}
    dx_t = -\alpha x_t dt + \sigma dW_t,
\end{equation}
where $W_t$ is a standard Wiener process. A total of 1000 trajectories were simulated over a time horizon $T_{\text{forward}} = 20$ with a time step $dt_{\text{forward}} = 0.01$, starting from initial conditions drawn from a standard normal distribution.

\begin{figure}[htbp]
    \centering
    \includegraphics[width=0.9\textwidth]{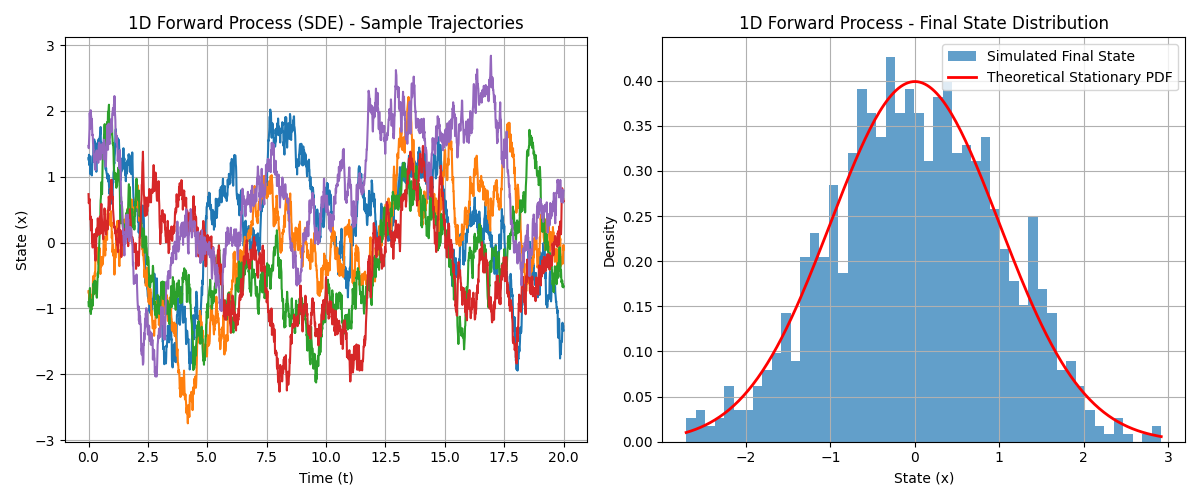}
    \caption{One-dimensional forward diffusion process (SDE). Left: Sample trajectories evolving over time. Right: Histogram of the final state distribution at $t=T_{\text{forward}}$ compared with the theoretical stationary distribution $\mathcal{N}(0, \sigma^2 / 2\alpha)$.}
    \label{fig:1d_forward}
\end{figure}

Figure \ref{fig:1d_forward} (left) displays several sample trajectories of this forward process, illustrating the convergence of the system state towards zero mean. The right panel of Figure \ref{fig:1d_forward} shows a histogram of the final states of all 1000 trajectories at $t = T_{\text{forward}}$. This empirical distribution closely matches the theoretical stationary distribution of the Ornstein-Uhlenbeck process, $\mathcal{N}(0, \sigma^2 / 2\alpha)$, confirming that the forward SDE successfully transforms the initial data distribution into a simple Gaussian prior.

The reverse-time generative process, according to the PH framework (Eq. (4.7)), is governed by the deterministic ODE:
\begin{equation}
    \dot{x}_t = (J - R - GG^\top) \nabla H(x_t) = -(\alpha + \sigma^2) x_t.
\end{equation}
This ODE was integrated from $t=0$ to $T_{\text{reverse}} = 10$ using the `RK45` method in SciPy, with 200 evaluation points. Several reverse trajectories were initiated from different starting points sampled from a uniform distribution spanning $[-5, 5]$.

\begin{figure}[htbp]
    \centering
    \includegraphics[width=0.9\textwidth]{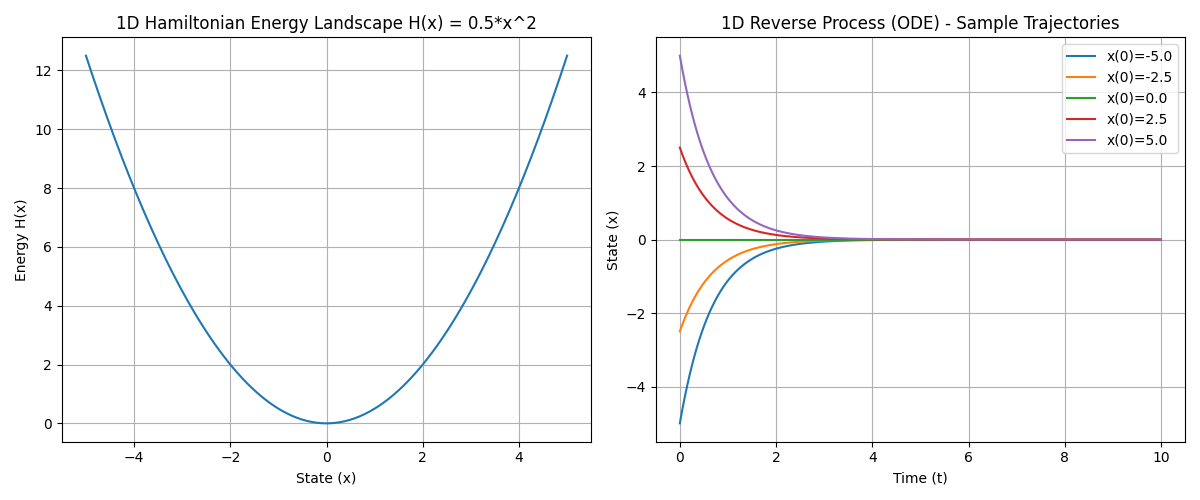}
    \caption{One-dimensional reverse generative process (ODE). Left: The quadratic Hamiltonian energy landscape $H(x) = \frac{1}{2}x^2$. Right: Sample reverse trajectories converging to the stable equilibrium at $x=0$, which is the minimum of the Hamiltonian.}
    \label{fig:1d_reverse}
\end{figure}

Figure \ref{fig:1d_reverse} (left) visualizes the quadratic Hamiltonian energy landscape $H(x)$. The right panel shows the sample reverse trajectories, all of which converge towards the single stable equilibrium at $x=0$, corresponding to the minimum of the Hamiltonian. This clearly demonstrates the energy-dissipative nature of the reverse PH dynamics, where the system flows downhill towards lower energy states, effectively generating samples from the target distribution concentrated at the origin in this simple case.

\subsection{Two-Dimensional Conceptual PH Diffusion Model}
\label{sec:num_2d_example}

To illustrate the applicability of the PH framework to more complex, multi-modal distributions, we consider a two-dimensional system. The Hamiltonian is defined as:
\begin{equation}
    H(\mathbf{x}) = (x_1^2 - 1)^2 + (x_2^2 - 1)^2,
\end{equation}
which possesses four minima at $(\pm 1, \pm 1)$, corresponding to four modes in the target data distribution. The gradient of this Hamiltonian is $\nabla H(\mathbf{x}) = [4x_1(x_1^2 - 1), 4x_2(x_2^2 - 1)]^\top$. The PH structural matrices were chosen as:
\begin{align}
    J &= \begin{bmatrix} 0 & -j \\ j & 0 \end{bmatrix} \quad \text{with } j = 0.5, \\
    R &= \begin{bmatrix} r_1 & 0 \\ 0 & r_2 \end{bmatrix} \quad \text{with } r_1 = r_2 = 0.2, \\
    G &= \begin{bmatrix} g_1 & 0 \\ 0 & g_2 \end{bmatrix} \quad \text{with } g_1 = g_2 = 1.0.
\end{align}
For the reverse process, the combined dissipation matrix $R + GG^\top$ is:
\begin{equation}
    R + GG^\top = \begin{bmatrix} r_1 + g_1^2 & 0 \\ 0 & r_2 + g_2^2 \end{bmatrix} = \begin{bmatrix} 1.2 & 0 \\ 0 & 1.2 \end{bmatrix}.
\end{equation}
This matrix is positive definite, satisfying the condition for global asymptotic stability of the reverse dynamics as per Theorem 5.1 in the main text.

\begin{figure}[htbp]
    \centering
    \includegraphics[width=\textwidth]{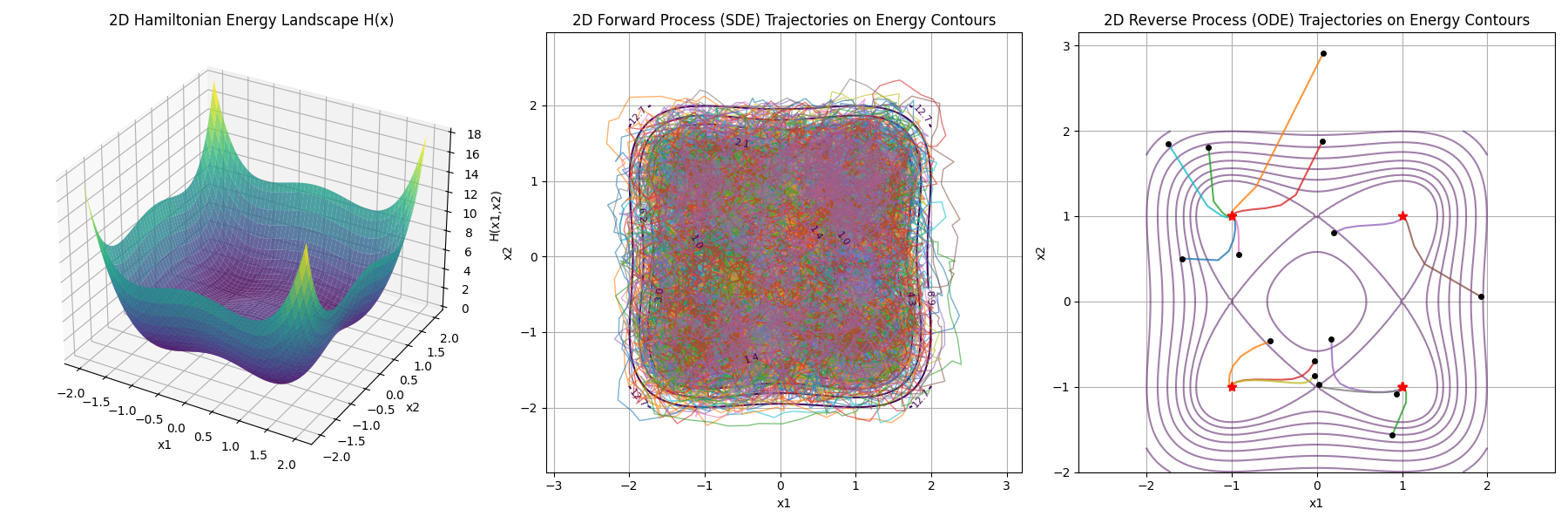}
    \caption{Two-dimensional PH diffusion model. Left: 3D surface plot of the multi-modal Hamiltonian energy landscape $H(\mathbf{x}) = (x_1^2 - 1)^2 + (x_2^2 - 1)^2$. Center: Contour plot of the energy landscape with trajectories of the forward SDE $d\mathbf{x}_t = (J-R)\nabla H(\mathbf{x}_t)dt + Gd\mathbf{W}_t$ starting near the origin and diffusing outwards. Right: Contour plot of the energy landscape with trajectories of the reverse ODE $\dot{\mathbf{x}}_t = (J-R-GG^\top)\nabla H(\mathbf{x}_t)$ (shown as solid lines) starting from random initial conditions (black circles) and converging to the minima of the Hamiltonian (red stars).}
    \label{fig:2d_ph}
\end{figure}

Figure \ref{fig:2d_ph} presents the results for the 2D simulation. The left panel shows the 3D surface plot of the chosen Hamiltonian, clearly exhibiting its four minima. The center panel displays the contour plot of this energy landscape overlaid with trajectories of the forward SDE, simulated for $T_{\text{forward}} = 50$ with $dt_{\text{forward}} = 0.01$. These trajectories, starting near the origin, diffuse according to the PH dynamics and injected noise, exploring the state space.

The right panel of Figure \ref{fig:2d_ph} illustrates the reverse generative process. Fifteen trajectories of the reverse ODE (Equation (4.7)), integrated from $t=0$ to $T_{\text{reverse}} = 15$, are plotted on the energy contours. Initial conditions (marked with black circles) were drawn from a Gaussian distribution $\mathcal{N}(0, 1.5^2 I)$. The trajectories (solid lines) flow towards the four minima of the Hamiltonian (marked with red stars), demonstrating the system's ability to generate samples corresponding to different modes of the target distribution. This visualization effectively captures the control-theoretic interpretation of the reverse process as a feedback-controlled system that actively dissipates energy to steer the state towards low-energy configurations.

These numerical examples collectively validate the core theoretical propositions of the PH diffusion framework. They demonstrate that the structured PH dynamics lead to stable and interpretable generative processes, both in simple unimodal and more complex multimodal settings. The energy-dissipative nature of the reverse process ensures convergence to the modes of the learned data distribution, as characterized by the minima of the Hamiltonian function.
\section{Conclusion and Positioning}

This paper develops a system-theoretic perspective on diffusion-based generative models
by embedding them within the framework of PH systems. By formulating
both the forward diffusion and the associated generative dynamics in terms of energy
functions, interconnection structure, and dissipation, we make explicit a set of
structural properties—passivity, energy balance, and Lyapunov stability—that are not
intrinsic to conventional score-based formulations.

A central theoretical contribution is the formal comparison between the proposed
PH reverse dynamics and the exact reverse-time stochastic differential equation associated
with the forward diffusion. This comparison establishes that the PH sampler coincides with
the probability-flow (zero-noise) limit of the true reverse SDE if and only if the learned
Hamiltonian exactly reproduces the true marginal densities. Outside this idealized setting,
the PH dynamics do not aim to recover the exact reverse diffusion, but instead define a
deterministic, dissipative flow whose equilibria coincide with stationary points of the
learned energy function.

This result clarifies the conceptual positioning of the proposed method. Rather than
serving as an exact probabilistic inverse of the forward diffusion, the PH formulation
prioritizes structural robustness over distributional exactness. In particular, the
closed-loop generative dynamics inherit passivity and Lyapunov stability by construction,
independently of score mismatch or modeling error. This stands in contrast to standard
reverse-time SDE samplers, whose stability and numerical behavior can degrade sharply when
the learned score is inaccurate.

More broadly, the port--Hamiltonian viewpoint reframes diffusion models as controlled
dynamical systems with explicit energy dissipation mechanisms. This perspective opens
several promising directions for future research, including structure-preserving neural
parameterizations of energy functions, robustness analysis under imperfect score learning,
and extensions to constrained, interconnected, or physics-informed generative systems.
By complementing probabilistic guarantees with system-theoretic structure, the proposed
framework provides a principled foundation for stable and interpretable generative
modeling.
\paragraph{Limitations.}
Several limitations of the proposed framework should be acknowledged.
First, the energy-based representation of the diffusion marginals is assumed rather than
derived, and approximation errors in the learned Hamiltonian may lead to discrepancies
between the induced port--Hamiltonian flow and the true reverse-time diffusion.
Second, the analysis focuses on frozen-time Lyapunov and passivity properties and does not
provide guarantees on the transient evolution of probability distributions, nor on exact
likelihood recovery.
Third, the deterministic PH reverse dynamics correspond to a probability-flow
approximation and therefore neglect stochastic effects that may be beneficial for
exploration and mode coverage in high-dimensional sampling.
Finally, the current formulation assumes globally Lipschitz energy gradients and
unconstrained state spaces; extending the theory to constrained systems, non-smooth
energies, or adaptive noise structures remains an open direction for future work.

\end{document}